\documentclass[10pt,usletter]{article}
\usepackage{natbib}
\usepackage[totalwidth=450pt,totalheight=640pt]{geometry}

\usepackage{times}
\usepackage{helvet}
\usepackage{courier}
\usepackage{balance}

\usepackage{amsthm,amsmath,amssymb,amstext}

\usepackage{tikz}
\usepackage{boxedminipage}
\usepackage{xspace}
\usepackage{url}

\newcommand{\SB}{\{\,}%
\newcommand{\SM}{\;{|}\;}%
\newcommand{\SE}{\,\}}%

\newcommand{\cO}{\mathcal{O}}

\newcommand{\NP}{\text{\normalfont NP}}

\newcommand{\W}[1][xxxx]{\text{\normalfont W}[#1]}

\newcommand{\fpt}{fixed-pa\-ra\-me\-ter trac\-ta\-ble\xspace}

\newcommand{\set}[1]{\left\{ #1 \right\}}
\newcommand{\myiff}{if and only if\xspace}

\newcommand{\dist}{{\normalfont \textsf{dist}}}

\newcommand{\LSVC}{\textsc{LS-Vertex Cover}\xspace}
\newcommand{\PLSVC}{\textsc{pLS-Vertex Cover}\xspace}

\newcommand{\HS}{\textsc{Hall Set}\xspace}

\newcommand{\CLIQUE}{\textsc{Clique}\xspace}

\newcommand{\citex}[1]{\citeauthor{#1}~(\citeyear{#1})}
\newcommand{\citey}[1]{\citeauthor{#1},~\citeyear{#1}}

\newtheorem{definition}{Definition}
 
\newtheorem{theorem}{Theorem}

\newtheorem{lemma}[theorem]{Lemma}

{}
\newtheorem{remark}{Remark}

\newcommand{\pbDefP}[4]{%
\noindent
\begin{center}
\begin{boxedminipage}{0.98 \columnwidth}
#1\\[5pt]
\begin{tabular}{l p{0.70 \columnwidth}}
Input: & #2\\
Parameter: & #3\\
Question: & #4
\end{tabular}
\end{boxedminipage}
\end{center}
}

\newcommand{\pbDefPtask}[4]{%
\noindent
\begin{center}
\begin{boxedminipage}{0.98 \columnwidth}
#1\\[5pt]
\begin{tabular}{l p{0.70 \columnwidth}}
Input: & #2\\
Parameter: & #3\\
Task: & #4
\end{tabular}
\end{boxedminipage}
\end{center}
}

\begin{document}

\title{Don't Be Strict in Local Search!}

\author{
Serge Gaspers \\
The University of New South Wales and\\
Vienna University of Technology\\
{gaspers@kr.tuwien.ac.at} \\
\and
Eun Jung Kim\\
LAMSADE-CNRS, Universit\'e Paris-Dauphine\\
{eunjungkim78@gmail.com} \\
\and
Sebastian Ordyniak \\
Vienna University of Technology \\
{ordyniak@kr.tuwien.ac.at} \\
\and
Saket Saurabh \\
The Institute of Mathematical Sciences\\
{saket@imsc.res.in}\\
\and
Stefan Szeider \\
Vienna University of Technology \\
{stefan@szeider.net}
}
 


\date{}

\nonfrenchspacing

\maketitle

\begin{abstract}
Local Search is one of the fundamental approaches to
combinatorial optimization and it is used throughout AI.
Several local search algorithms are based on searching
the $k$-exchange neighborhood. This is the
set of solutions that can be obtained from the current
solution by exchanging at most $k$ elements.
As a rule of thumb, the larger $k$ is, the better
are the chances of finding an improved solution.
However, for inputs of size $n$, a na\"ive brute-force
search of the $k$-exchange neighborhood requires
$n^{O(k)}$ time, which is not practical even for very
small values of $k$.

Fellows et al.\ (IJCAI 2009) studied whether this brute-force
search is avoidable and gave positive and negative answers
for several combinatorial problems.
They used the notion of local search in a strict sense. That is, 
an improved solution needs to be found in the $k$-exchange
neighborhood even if a global optimum can be found efficiently.

In this paper we consider a natural relaxation of local search,
called \emph{permissive} local search (Marx and Schlotter, IWPEC 2009) and investigate whether
it enhances the domain of tractable inputs.
We exemplify this approach on a fundamental combinatorial problem,
\textsc{Vertex Cover}. More precisely, we show that for a class of inputs,
finding an optimum is hard, strict local search is hard, but permissive local
search is tractable.

We carry out this investigation in the framework of
parameterized complexity.
\end{abstract}


\section{Introduction}

Local search is one of the most common approaches applied in practice
to solve hard optimization problems.  It is used as a subroutine in
several kinds of heuristics, such as evolutionary algorithms and
hybrid heuristics that combine local search and genetic
algorithms. The history of employing local search in combinatorial
optimization and operations research dates back to the 1950s with the
first edge-exchange algorithms for the traveling salesperson~\citep{Bock58,Cores58}.

In general, such algorithms start from a feasible solution and
iteratively try to improve the current solution. Local search
algorithms, also known as neighborhood search algorithms, form a large
class of improvement algorithms. To perform local search, a problem
specific neighborhood distance function is defined on the solution
space and a better solution is searched in the neighborhood of the
current solution. In particular, many local search algorithms are
based on searching the \emph{$k$-exchange neighborhood}. This is the set
of solutions that can be obtained from the current solution by
exchanging at most $k$ elements. 


Most of the literature on local search is primarily devoted to
experimental studies of different heuristics. The theoretical study of
local search has developed mainly in four directions.

The first direction is the study of performance guarantees of local
search, i.e., the quality of the
solution~\citep{Alimonti95,Alimonti97,GuptaTardos00,KhannaMotwaniSudanVazirani98,PapadimitriouSteiglit77}. The
second direction of the theoretical work is on the asymptotic
convergence of local search in probabilistic settings, such as simulated
annealing~\citep{AartsKorstLaarhoven97}. The third direction
concerns the time required to reach a local optimum.  The fourth
direction is concerned with so-called kernelization techniques
\citep{GuoNiedermeier07} for local search, and aims at providing the
basis for putting our theoretical results to work in practice.

In a recent paper by \citex{FellowsRosamondFominLokshtanovSaurabhVillanger09} another twist in the study of local
search has been taken with the goal of answering the following natural question.
Is there a faster way of searching the $k$-exchange neighborhood
than brute-force? This question is important because
the typical running time of a brute-force algorithm is $n^{O(k)}$,
where $n$ is the input length. Such a running time becomes a real obstacle
in using $k$-exchange neighborhoods in practice even for very
small values of $k$. For many years most algorithms
searching an improved solution in the $k$-exchange neighborhood
had an $n^{O(k)}$ running time, creating the impression
that this cannot be done significantly faster than
brute-force search. But is there mathematical evidence for
this common belief? Or is it possible for some problems
to search $k$-exchange neighborhoods in time $O(f(k)n^c)$,
where $c$ is a small constant, which can make local search
much more powerful?

An appropriate tool to answer all these questions is \emph{parameterized
complexity}. In the parameterized complexity framework, for decision
problems with input size $n$, and a parameter $k$, the goal is to
design an algorithm with running time $f(k)n^{\cO(1)}$ , where $f$ is
a function of $k$ alone. Problems having such an algorithm are said to
be \emph{fixed parameter tractable (FPT)}. There is also a theory of hardness
that allows us to identify parameterized problems that are not
amenable to such algorithms. The hardness hierarchy is represented by
$W[i]$ for $i\geq 1$. The theory of parameterized complexity was developed by 
\citex{DowneyFellows99}. For recent developments, see the book by
\citex{FlumGrohe06}.

In this paper we consider two variants of the local search problem for
the well-known \textsc{Vertex Cover} problem, that is, the \emph{strict} and the
\emph{permissive} variant of local
search~(Marx and Schlotter 2011; Krokhin and Marx). 
In the strict variant the task is to either determine that there is no better solution in
the $k$-exchange neighborhood, or to find a better solution in the $k$-exchange neighborhood.
In the permissive variant, however, the task is to either determine that there is no better solution in
the $k$-exchange neighborhood, or to find a better solution, which may
or may not belong to the $k$-exchange neighborhood.
Thus, permissive local search does not require the improved solution to belong to the
local neighborhood, but still requires that at least the local neighborhood
has been searched before abandoning the search.
It can
therefore be seen as a natural relaxation of strict local search with the
potential to make local search applicable to a wider range of problems or instances.
Indeed, we will present a class of instances for \textsc{Vertex Cover} where
strict local search is $\W[1]$-hard, but permissive loal search is FPT.

In heuristic local search, there is an abundance of techniques, such as random restarts and
large neighborhood search, to escape local minima and boost the performance of algorithms \citep{HoosStuetzle04}.
Permissive local search is a specific way to escape the strictness of local search, but
allows a rigorous analysis and performance guarantees.

\vspace{-2pt} 

\paragraph{Relevant results.} 
Recently, the parameterized complexity of local search has gained more
and more attention. Starting with the first breakthrough in this area
by~\citex{Marx08} who investigated the parameterized complexity of TSP,
several positive and negative results have been obtained in many areas
of AI. For instance, the local search problem has already been
investigated for a variant of the feedback edge set
problem~\citep{KhullerBhatiaPless03}, for the problem of finding a
minimum weight assignment for a Boolean constraint satisfaction
instance~(Krokhin and Marx), for the stable marriage problem with
ties~\citep{MarxSchlotter11}, for combinatorial problems on
graphs~\citep{FellowsRosamondFominLokshtanovSaurabhVillanger09}, for
feedback arc set problem on
tournaments~\citep{FominLokshtanovRamanSaurabh10}, for the satisfiability
problem~\citep{Szeider11b}, and for Bayesian network structure
learning~\citep{OrdyniakSzeider10}.

\paragraph{Our results.}
We investigate local search for the fundamental \textsc{Vertex Cover}
problem. This well-known combinatorial optimization problem
has many applications \citep{AbuKhzamCFLSS04,GomesMPV06}
and is closely related to two other classic problems,
\textsc{Independent Set} and \textsc{Clique}.
All our results for \textsc{Vertex Cover} also hold for
the \textsc{Independent Set} problem, and for the \textsc{Clique} problem on
the complement graph classes.
\begin{itemize}
\item We give the first compelling evidence that it is possible to
  enhance the tractability of local search problems if 
  permissive local search is considered instead of strict local search.
  Indeed, the
  permissive variant allows us to solve the local search problem for
  \textsc{Vertex Cover}
  for a significantly larger class of sparse graphs than strict local
  search. 
\item We show that the strict local search \textsc{Vertex Cover} problem
  remains $W[1]$-hard for special sparse instances, improving a result
  from~\citex{FellowsRosamondFominLokshtanovSaurabhVillanger09}. On the
  way to this result we introduce a size-restricted version of a Hall
  set problem which be believe to be interesting in its own right.
\item We answer a question of Krokhin and Marx in the affirmative, who
  asked whether there was a problem where finding the optimum is hard, 
  strict local search is hard, but permissive local search is FPT.
\end{itemize}

\section{Preliminaries}
The \emph{distance} between two sets $S_1$ and $S_2$
is $\dist(S_1,S_2)=|S_1\cup S_2| - |S_1 \cap S_2|$.
We say that $S_1$ is in the \emph{$k$-exchange neighborhood} of $S_2$
if $\dist(S_1,S_2)\le k$.
If we consider a universe $V$ with $S_1,S_2\subseteq V$, the characteristic
functions of $S_1$ and $S_2$ with respect to $V$ are at Hamming distance
at most $k$ if $\dist(S_1,S_2)\le k$.

All graphs considered in this paper are finite, undirected, and simple.
Let $G=(V,E)$ be a graph, $S\subseteq V$ be a vertex set, and $u,v\in V$ be vertices.
The \emph{distance} $\dist(u,v)$ between $u$ and $v$ is the minimum number of
edges on a path from $u$ to $v$ in $G$.
The \emph{(open) neighborhood} of $v$ is $N(v)= \SB u\in V \SM  uv\in E \SE$,
i.e., the vertices at distance one from $v$,
and its \emph{closed neighborhood} is $N[v]=N(v)\cup \set{v}$.
We also define $N(S) = \bigcup_{u\in S} N(u) \setminus S$ and
$N[S] = N(S) \cup S$.
More generally, $N^d(S)$ and $N^d[S]$ denote the set of vertices
at distance $d$ and at distance at most $d$ from a vertex in $S$,
respectively. We write $N^d(v)$ and $N^d[v]$ for $N^d(\set{v})$ and $N^d[\set{v}]$, respectively.
The \emph{degree} of $v$ is $d(v)=|N(v)|$.
These notations may be subscripted by $G$, especially if the graph is not clear from the context.

The graph $G\setminus S$ is obtained from $G$ by removing all vertices
in $S$ and all edges incident to vertices in $S$.
The subgraph of $G$ \emph{induced} by $S$ is $G\setminus (V\setminus S)$ and it is
denoted $G[S]$.
The set $S$ is a \emph{vertex cover} of $G$ if $G\setminus S$ has no edge.
The set $S$ is an \emph{independent set} of $G$ if $G[S]$ has no edge.
The graph $G$ is \emph{bipartite} if its vertex set can be partitioned into two
independent sets $A$ and $B$. In this case, we also denote the graph by
a triple $G=(A,B,E)$.

The instances considered in this paper are $d$-degenerate graphs.

The \emph{degeneracy} of $G$ is the minimum $d$ such that every subgraph of $G$ has a vertex of degree
at most $d$.
Degeneracy is a fundamental sparsity measure of graphs.
A graph $G'$ is obtained from $G$ by \emph{subdividing} an edge $xy\in E$ if $G'$ is obtained
by removing the edge $xy$, and adding a new vertex $z_{xy}$ and edges $x z_{xy}$ and
$z_{xy} y$.
A graph $G'$ is obtained from $G$ by \emph{subdividing} an edge $xy\in E$ \emph{twice} if $G'$ is obtained
by removing the edge $xy$, and adding new vertices $z_{xy}$ and $z'_{xy}$ and edges $x z_{xy}$, $z_{xy} z'_{xy}$ and
$z'_{xy} y$.
The graph $G$ is \emph{$2$-subdivided} if $G$ can be obtained from a graph $G'$ by subdividing
each edge of $G'$ twice.


\section{Hardness proofs}


In this section we show that strict local search for \textsc{Vertex Cover} is  \W[1]-hard on 2-subdivided graphs.

 \pbDefP{\LSVC}
 {A graph $G=(V,E)$, a vertex cover $S\subseteq V$ of $G$, and an integer $k$.}
 {The integer $k$.}
 {Is there a vertex cover $S' \subseteq V$ in the $k$-exchange neighborhood of $S$ with $|S'|<|S|$?}



Our proof will strengthen the following result of 
\citex{FellowsRosamondFominLokshtanovSaurabhVillanger09}.

\begin{theorem}[\citey{FellowsRosamondFominLokshtanovSaurabhVillanger09}]
 \LSVC is {\rm \W[1]}-hard and remains {\rm \W[1]}-hard when restricted to 3-degenerate graphs.
\end{theorem}

As 2-subdivided graphs are 2-degenerate, our result implies that \LSVC is \W[1]-hard when restricted to
2-degenerate graphs as well.

We first show that the following intermediate problem is \W[1]-hard for 2-subdivided graphs.

 \pbDefP{\HS}
 {A bipartite graph $G=(A,B,E)$ and an integer $k$.}
 {The integer $k$}
 {Is there a set $S \subseteq A$ of size at most $k$ such that $|N(S)|<|S|$?}

As \HS is a very natural problem related to matching theory,
and to give an intuition for the \W[1]-hardness proof for \HS restricted to 2-subdivided graphs,
we first show that \HS is \W[1]-hard on general graphs.

\tikzset{vertex/.style={inner sep=.12em,circle,fill=black,draw},
         label distance=-2pt}

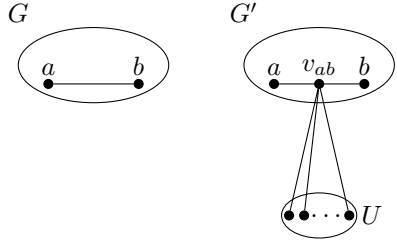
\begin{figure}[tb]
  \centering
  \begin{tikzpicture}
    \draw (0,0) ellipse (1cm and .5cm);
    \node at (-1,0.7) {$G$};
    \node (a) at (-0.6,-0.25) [vertex,label=above:$a$] {};
    \node (b) at (0.6,-0.25) [vertex,label=above:$b$] {};
    \draw (a)--(b);
    
    \draw (3,0) ellipse (1cm and .5cm);
    \node at (2,0.7) {$G'$};
    \node (ap) at (2.4,-0.25) [vertex,label=above:$a$] {};
    \node (m) at (3,-0.25) [vertex,label=above:$v_{ab}$] {};
    \node (bp) at (3.6,-0.25) [vertex,label=above:$b$] {};
    \draw (ap)--(bp);
    
    \draw (3,-2) ellipse (.5cm and .3cm);
    \node at (3.7,-2) {$U$};
    \node (u1) at (2.6,-2) [vertex] {};
    \node (u2) at (2.8,-2) [vertex] {};
    \node at (3.12,-2) {$\dots$};
    \node (u3) at (3.4,-2) [vertex] {};
    \draw (u1)--(m)--(u2) (m)--(u3);
  \end{tikzpicture}
  \caption{\label{fig:h1} Reduction from Lemma \ref{lem:h1} illustrated for one edge of $G$.}
\end{figure}

\begin{lemma}\label{lem:h1}
 \HS is {\rm \W[1]}-hard.
\end{lemma}
\begin{proof}
 We prove the lemma by a parameterized reduction from \CLIQUE, which is \W[1]-hard \citep{DowneyFellows99}.

 \pbDefP{\CLIQUE}
 {A graph $G$ and an integer $k$.}
 {The integer $k$.}
 {Does $G$ have a clique of size $k$?}

Let $(G,k)$ be an instance for \CLIQUE. We construct an instance $(G',k')$ for \HS as follows.
Set $k' := \binom{k}{2}$.
Subdivide each edge $e$ of $G$ by a new vertex $v_e$, then add a set of $t$ new vertices $U=\{u_1, \dots, u_t\}$,
with $t := k'-k-1$, and add an edge $v_eu$ for each $e\in E$ and $u\in
U$. Set $A:=\SB v_e \SM e\in E \SE$ and $B:=V\cup U$.

Suppose $G$ has a clique $C$ of size $k$. Consider the set $S := \SB
v_e\in A \SM e \subseteq C\SE$, i.e., the set of vertices introduced
in $G'$ to subdivide the edges of $C$. Then, $|S| = \binom{k}{2} = k'$. Moreover, $|N(S)| = |C \cup U| = k+t = k'-1$. Thus, $S$ is a
Hall set of size $k'$.

On the other hand, suppose $S \subseteq A$ is a Hall Set of size at most $k'$. As $S \neq \emptyset$, we have that $U \subseteq N(S)$.
Since each vertex from $S$ has two neighbors in $V$, we have $|S| \le \binom{|V\cap N(S)|}{2}$.
From $|S|>t+|V\cap N(S)|$ it follows that $|S|-\binom{k}{2}+k+1 > |V\cap N(S)|$, which can only be achieved if
$|S|=\binom{k}{2}$ and $|V\cap N(S)|=k$.
But then, $V\cap N(S)$ is a clique of size $k$ in $G$.
\end{proof}

We now generalize the above proof and reduce \CLIQUE to \HS restricted to
2-subdivided graphs.

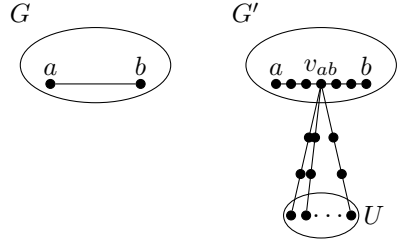
\begin{figure}[tb]
  \centering
  \begin{tikzpicture}
    \draw (0,0) ellipse (1cm and .5cm);
    \node at (-1,0.7) {$G$};
    \node (a) at (-0.6,-0.25) [vertex,label=above:$a$] {};
    \node (b) at (0.6,-0.25) [vertex,label=above:$b$] {};
    \draw (a)--(b);
    
    \draw (3,0) ellipse (1cm and .5cm);
    \node at (2,0.7) {$G'$};
    \node (ap) at (2.4,-0.25) [vertex,label=above:$a$] {};
    \node at (2.6,-0.25) [vertex] {};
    \node at (2.8,-0.25) [vertex] {};
    \node (m) at (3,-0.25) [vertex,label=above:$v_{ab}$] {};
    \node at (3.2,-0.25) [vertex] {};
    \node at (3.4,-0.25) [vertex] {};
    \node (bp) at (3.6,-0.25) [vertex,label=above:$b$] {};
    \draw (ap)--(bp);
    
    \draw (3,-2) ellipse (.5cm and .3cm);
    \node at (3.7,-2) {$U$};
    \node (u1) at (2.6,-2) [vertex] {};
    \node (u2) at (2.8,-2) [vertex] {};
    \node at (3.12,-2) {$\dots$};
    \node (u3) at (3.4,-2) [vertex] {};
    \draw (u1)--(m) node[pos=0.6,vertex] {} node[pos=0.3,vertex] {};
    \draw (u2)--(m) node[pos=0.6,vertex] {} node[pos=0.3,vertex] {};
    \draw (u3)--(m) node[pos=0.6,vertex] {} node[pos=0.3,vertex] {};
  \end{tikzpicture}
  \caption{\label{fig:h2} Reduction from Lemma \ref{lem:h2} illustrated for one edge of $G$.}
\end{figure}

\begin{lemma}\label{lem:h2}
 \HS is {\rm \W[1]}-hard even if restricted to $2$-subdivided graphs. 
\end{lemma}
\begin{proof}
%
%
Let $(G,k)$ be an instance for \CLIQUE. We construct an instance $(G',k')$ for \HS as follows.
Set $t=\binom{k}{2}-k-1$ and $k'=(3+t) \cdot \binom{k}{2}$.
Subdivide each edge $e$ of $G$ by a new vertex $v_e$. Then add a set of $t$ new vertices $U=\{u_1, \dots, u_t\}$,
and add an edge $v_eu$ for each $e\in E$ and $u\in U$. This graph is bipartite with bipartition $(A,B)$ where
$A:=\SB v_e \SM e\in E\SE$ and $B:=V\cup U$. Now, make a 2-subdivision of each edge.
Choose $A'\supseteq A$ and $B'\supseteq B$ so that $(A',B')$ is a bipartition of the vertex set of the resulting graph $G'=(V',E')$.

Suppose $G$ has a clique $C$ of size $k$. Consider the set $S := \SB
v_e\in A \SM e \subseteq C \SE$, i.e., the set of vertices introduced
to subdivide the edges of $C$. Set $S' := S\cup N^2_{G'}(S)$. Then, $S' \subseteq A'$ and $|S'| = (3+t)\binom{k}{2} = k'$.
Moreover, $|N(S')| = |C \cup N_{G'}(S) \cup U| = k+(2+t) \binom{k}{2}+t = k'-1$. Thus, $S'$ is a
Hall set of size $k'$.

\smallskip

On the other hand, suppose $S \subseteq A$ is a Hall Set of $G'$ of size at most $k'$.
Set $S':=S$ and exhaustively apply the following rule.

\begin{description}
\item[Minimize] If there is a vertex $v\in A$ such that $\emptyset \neq S' \cap (\{v\}\cup N^2(v)) \neq \{v\}\cup N^2(v)$,
 then remove $\{v\}\cup N^2(v)$ from $S'$.
\end{description}

To see that the resulting $S'$ is a Hall Set in $G'$, consider a set $S_2$ that is obtained from a Hall Set $S_1$ by one application
of the \textbf{Minimize} rule. Suppose $v\in A$ such that $\emptyset \neq S_1 \cap (\{v\}\cup N^2(v)) \neq \{v\}\cup N^2(v)$ but
$S_2 \cap (\{v\}\cup N^2(v)) = \emptyset$. If $v\notin S_1$, then removing a vertex $u\in N^2(v)$ from $S_1$
decreases $|S_1|$ by one and $|N(S_1)|$ by at least one; namely the vertex in $N(u)\cap N(v)$ disappears from
$N(S_1)$ when removing $v$ from $S_1$. After removing all vertices in $N^2(v)$ from $S_1$, we obtain a set $S_2$
such that $|N(S_1)|-|N(S_2)| \ge |S_1|-|S_2|$. As $|N(S_1)|\le |S_1|-1$, we obtain that $|N(S_2)| \le |S_2|-1$ and therefore,
$S_2$ is a Hall Set of size at most $k'$. On the other hand, if $v\in S_1$, then there is a vertex $u\in N^2(v)\setminus S_1$.
Then, removing $v$ from $S_1$ decreases $|S_1|$ by one and $|N(S_1)|$ by at least one; namely, the vertex in
$N(u)\cap N(v)$ disappears from $N(S_1)$ when removing $v$ from $S_1$. Thus, $S_1\setminus \{v\}$ is a Hall Set
of size at most $k'$ and we can appeal to the previous case with $S_1:=S_1\setminus \{v\}$.

We have now obtained a Hall set $S'$ of size at most $k'$ such that for every $v\in A$, either $S' \cap (\{v\}\cup N^2(v)) = \emptyset$
or $\{v\}\cup N^2(v) \subseteq S'$

As $S' \neq \emptyset$, we have that $U \subseteq N(S')$.
Expressing the size of $S'$ in terms of vertices from $A$ we obtain that
\begin{align*}
|S'|=(3+t) |S'\cap A|.
\end{align*}
Similarly, we express $|N(S')|$, which contains all vertices from $N(S'\cap A)$, all vertices in $U$ and some vertices from $V\cap N^3(S')$.
\begin{align*}
|N(S')| = (2+t) |S'\cap A| + t + |V\cap N^3(S')|.
\end{align*}
Now,
\begin{align*}
|S'| &> |N(S')|\\
(3+t) |S'\cap A| &> (2+t) |S'\cap A| + t + |V\cap N^3(S')|\\
|S'\cap A| - \binom{k}{2} &> |V\cap N^3(S')|-k-1.
\end{align*}
Since there are two vertices in $V\cap N^3(v)$ for each vertex $v\in S'\cap A$, we have $|S'\cap A| \le \binom{|V\cap N^3(S')|}{2}$.
Moreover, $|S'\cap A|\le \frac{k'}{3+t} = \binom{k}{2}$.
Therefore, the previous inequality can only be satisfied if $|S'\cap A| = \binom{k}{2}$ and $|V\cap N^3(S')|=k$.
Then, $|V\cap N^3(S')|$ is a clique of size $k$ in $G$.
\end{proof}

Finally, we rely on the previous lemma to establish
\W [1]-hardness of \LSVC for 2-subdivided graphs.
The reduction will make clear that the \HS problem captures
the essence of the \LSVC problem.

\begin{theorem}\label{thm:hardnessVC}
 \LSVC is {\rm \W[1]}-hard when restricted to 2-subdivided graphs.
\end{theorem}
\begin{proof}
%
The proof uses a reduction from \HS restricted to 2-subdivided graphs.
Let $(G,k)$ be an instance for \HS where $G=(A,B,E)$ is a 2-subdivided graph.
The set $A$ is a vertex cover for $G$. Consider $(G,A,k')$, with $k':=2k-1$ as an instance for \LSVC.

Let $S \subseteq A$ be a Hall Set of size at most $k$ for $G$, i.e., $|N(S)|<|S|$. Then, $(A \setminus S) \cup N(S)$ is
a vertex cover for $G$ of size at most $|A|-1$. Moreover, this vertex cover is in the $k'$-exchange neighborhood of
$A$.

On the other hand, let $C$ be a vertex cover in the $k'$-exchange neighborhood of $A$ such that $|C|<|A|$.
Set $C':=C$.
If $|A \setminus C| > k$, then add $|A \setminus C| - k$ vertices from $A\setminus C$ to $C'$. The resulting set $C'$ is also
in the $k'$-exchange neighborhood of $A$ and $|C'|<|A|$. Set $S:=A\setminus C'$. Then, $|S|\le k$.
As $C'$ is a vertex cover, $N(S) \subseteq C'$. But since $C'$ is smaller than $A$, we have that $|C'\cap B|\le |S|-1$
Therefore, $|N(S)|\le |C'\cap B|\le |S|-1$, which shows that $S$ is a Hall Set of size at most $k$ for $G$.
\end{proof}

\section{FPT Algorithm}

In this section we will show that the permissive version of \LSVC is \fpt for a generalization of 2-subdivided graphs.

\pbDefPtask{\PLSVC}{A graph $G$, a vertex cover $S$, and a positive integer
  $k$.}{The integer $k$.}
{Determine that $G$ has no vertex cover $S'$ with $\dist(S,S')\leq k$ and $|S'|<|S|$ or find a vertex cover $S''$ with $|S''|<|S|$.}

The initial algorithm will be randomized, and we will exploit the following pseudo-random object and theorem to derandomize it.
\begin{definition}[\citey{NaorSS95}]
An \emph{$(n,t)$-universal set} $\mathcal F$ is a set of functions from
$\{1,\ldots,n\}$ to $\{0,1\}$, such that for every subset $S\subseteq
\{1,\ldots,n\}$ with $|S|=t$, the set 
$\mathcal F|_{S}=\SB f|_S \SM f\in {\cal F}\SE$ 
is equal to the set $2^S$ of all the functions from $S$ to $\{0,1\}$.
\end{definition}

\begin{theorem}[\citey{NaorSS95}]
\label{propuniversalsets}
There is a deterministic algorithm with running time $O(2^t t^{O(\log t)} n\log n)$ that constructs an $(n,t)$-universal set $\mathcal F$ such that $|\mathcal F|=2^t t^{O(\log t)}\log n$.
\end{theorem}

Our FPT algorithm will take as input a $\beta$-separable graph.

\begin{definition}
For a fixed non-negative integer $\beta$, a graph $G=(V,E)$ is \emph{$\beta$-separable}  
if there exists a bipartition of $V$ into $V_1$ and $V_2$ such that
\begin{itemize}
\item for each $v\in V_1$, $|N(v)\cap V_1|\leq \beta$, and
\item for each $w \in V_2$, $|N(w)|\leq \beta$. 
\end{itemize}
A bipartition of $V$ satisfying these properties is a partition \emph{certifying $\beta$-separability}.
By ${\cal G}(\beta)$ we denote the set of all  $\beta$-separable graphs.
\end{definition}
\begin{remark}
{\rm Observe that a graph of degree at most $d$ is $d$-separable. Similarly every $2$-subdivided graph
is $2$-separable.}
\end{remark}

The following lemma characterizes solutions for \PLSVC that belong to the $k$-exchange neighborhood of $S$.

\begin{lemma}
\label{lem:structurallemma}
Let $G=(V,E)$ be a graph, $S$ be a vertex cover of $G$ and $k$ be a positive integer. Then there exists a 
vertex cover $S'$ such that $|S'|<|S|$ and $\dist(S,S')\leq k$ if and only if there exists a set $S^*\subseteq S$ such that

\smallskip
\noindent
1.~~$S^*$ is an independent set,\\
2.~~$|N(S^*)\setminus S|< |S^*|$, and\\ 
3.~~$|N(S^*)\setminus S|+|S^*|\leq k$. 
\end{lemma}
\begin{proof}
We first show the forward direction of the proof. Let $S^*=S\setminus S'$. Since $I=V\setminus S'$ is an independent set  
and $S^*\subseteq I $ we have that $S^*$ is an independent set.  Furthermore, since $S^*$ is in $I$ we have that $N(S^*)\subseteq S'$ 
and in particular $N(S^*)\setminus S$ is the set of vertices that are present in $S'$ but not in $S$. Since $|S'|<|S|$ we have that $|N(S^*)\setminus S|< |S^*|$ and by the fact that $\dist(S,S')\leq k$ we have that $|N(S^*)\setminus S|+|S^*|\leq k$. For the reverse direction it is easy to see that 
$(S\setminus S^*) \cup (N(S^*)\setminus S)$ is the desired $S'$. This completes the proof. 
\end{proof}

To obtain the FPT algorithm for  \PLSVC  on ${\cal G}(\beta)$ we will use Lemma~\ref{lem:structurallemma}. More 
precisely, our strategy is to obtain an FPT algorithm for  finding a subset $Q\subseteq S$ such that $Q$ is an independent set and 
$S^* \subseteq Q $. Here, $S^*$ is as described in Lemma~\ref{lem:structurallemma}. Thus, our main technical lemma is 
the following.

\begin{lemma}
\label{lem:mainlemma}
Let $\beta$ be a fixed non-negative integer.
Let $G$ be a $\beta$-separable graph, $S$ be a vertex cover of $G$ and $k$ be a positive integer.
There is a $O(2^q q^{O(\log q)} n \log n)$ time algorithm finding a family ${\cal Q}$ of 
subsets of $S$ such that (a) $|{\cal Q}|\leq 2^q q^{O(\log q)} \log n$, (b) each $Q\in {\cal Q}$ is an independent set, and (c) 
if there exists a $S^*$ as
described in Lemma~\ref{lem:structurallemma}, then there exists a $Q\in \cal Q$ such that $S^* \subseteq Q $. Here, $q=k+\beta k$.
\end{lemma}
We postpone the proof of Lemma~\ref{lem:mainlemma} and first give the main result that uses Lemma~\ref{lem:mainlemma} crucially. 

\begin{theorem}
\label{thm:fptplsvc}
Let $\beta$ be a fixed non-negative integer. 
\PLSVC  is FPT on ${\cal G}(\beta)$ with an algorithm running in time $2^q q^{O(\log q)}n^{O(1)}$, where $q=k+\beta k$. 
\end{theorem}
\begin{proof}
Let $G$ be the input graph from ${\cal G}(\beta)$, $S$ be a vertex cover of $G$, and $k$ be a positive integer. 
Fix $q=k+\beta k$ and $I=V\setminus S$. We first apply Lemma~\ref{lem:mainlemma} and obtain a family ${\cal Q}$ of 
subsets of $S$ such that (a) $|{\cal Q}|\leq 2^q q^{O(\log q)} \log n$ and (b) each $Q\in {\cal Q}$ is an independent set. The family $Q$ has the additional 
property that if there exists a set $S^*$ as described in Lemma~\ref{lem:structurallemma}, then 
there exists a $Q\in \cal Q$ such that $S^* \subseteq Q $. 

For every $Q\in \cal Q$, the algorithm proceeds as follows. Consider the bipartite graph $G[Q\cup I]$. Now in polynomial time
check whether there exists a subset $W\subseteq Q$ such that $|N(W)|<|W|$ in $G[Q\cup I]$. This is done by checking
Halls' condition that says that  there exists a matching saturating $Q$ if and only if for all $A\subseteq Q$, $|N(A)|\geq |A|$. 
A polynomial time algorithm that finds a maximum matching in a bipartite graph can be used to find a violating set $A$ if there exists one. 
See~\citex{Kozen91} for more details.  Returning to our algorithm, if we find such set 
$W$ then we return $S'=(S\setminus W) \cup N(W)$. Clearly, $S'$ is a vertex cover and $|S'|<|S|$. 
Now we argue that if for every $Q\in \cal Q$ we do not obtain the desired $W$, then 
there is no vertex cover  $S'$ such that $|S'|<|S|$ and $\dist(S,S')\leq k$. However, this is guaranteed by the fact that 
if there would exist such a set $S'$, then by Lemma~\ref{lem:structurallemma} there exist a desired $S^*$. Thus, when we consider 
the set $Q\in \cal Q$ such that $S^*\subseteq Q$ then we would have found a $W\subseteq Q$ such that $|N(W)|<|W|$ in $G[Q\cup I]$. 
This proves the correctness of the algorithm. The running time of the algorithm is governed by the size of the family $\cal Q$.  
This completes the proof.
\end{proof}

To complete the proof of Theorem~\ref{thm:fptplsvc}, the only remaining component is a proof of Lemma~\ref{lem:mainlemma} which 
we give below.
\begin{proof}[Proof of Lemma~\ref{lem:mainlemma}]
Let $G$ be a $\beta$-separable graph, $S$ be a vertex cover of $G$, and $k$ be a positive integer. 
By the proof of Lemma~\ref{lem:structurallemma} we know that if there exists a vertex cover
$S'$ such that $|S'|<|S|$ and $\dist(S,S')\leq k$ then there exists a set
$S^*\subseteq S$ such that 

\medskip
\noindent
1.~~$S^*$ is an independent set,\\
2.~~$|N(S^*)\setminus S|< |S^*|$, and\\
3.~~$|N(S^*)\setminus S|+|S^*|\leq k$.
\medskip

We first give a randomized procedure that produces a family $\cal Q$ satisfying the properties of the lemma with high probability.
In a second stage, we will derandomize
it using 
universal sets.
For our argument we fix one such $S^*$ and let $V_1$ and $V_2$ be a partition certifying $\beta$-separability of $G$. 
Let $S_1=S^*\cap V_1$ and $S_2=S^*\cap V_2$. 
Since $G$ is a $\beta$-separable graph, we have that $|N[S_1]\cap (V_1\cap S)|+|N[S_2]\cap S| \leq \beta |S_1| + |S_1| + \beta |S_2| + |S_2| \leq k+\beta k$. We also know that 
$|S^*|\leq k$. Let $q=k+\beta k$ and $A=(N[S_1]\cap (V_1\cap S))\cup (N[S_2]\cap S)$.  
Now, uniformly at random color the vertices of $S$ with $\{0,1\}$, that is, color each vertex of $S$ 
with $0$ with probability $\frac{1}{2}$ and with $1$ otherwise. Call this coloring $f$. 
The probability that for all $x \in S^*$, $f(x)=0$ and for all $y\in (A\setminus S^*)$, $f(y)=1$, is 
\[ \frac{1}{2^{|A|}} \geq  \frac{1}{2^q}.\] 
Given the random coloring $f$ we obtain a set $Q(f)\subseteq S$ with the following properties 
\begin{itemize}
\item  $Q(f)$ is an independent set; and 
\item  with probability at least $2^{-q}$, $S^*\subseteq Q(f)$. 
\end{itemize}
We obtain the set $Q(f)$ as follows. 
\begin{quote}
Let $C_0=\SB v \SM v\in S,~f(v)=0\SE$, that is, $C_0$ contains all the vertices of 
$S$ that have been assigned $0$ by $f$.  
Let $C_0^1 \subseteq C_0\cap V_2$ be the set of vertices that have degree at least 1
in $G[C_0]$. Let $C_0':=C_0\setminus C_0^1$. 
Let $E_0'$ be the set of edged in the induced 
graph $G[C_0']$ and $V(E_0')$ be the set of end-points of the edges in $E_0'$. 
Define $Q(f):=C_0'\setminus V(E_0')$. 
\end{quote}


By the procedure it is clear that $Q(f)$ is an independent set. However, note that it is possible that $Q(f)=\emptyset$. 
Now we show that with probability at least $e^{-q}$, $S^*\subseteq
Q(f)$. Let $C_i=\SB v \SM v\in S,~f(v)=i\SE$, $i\in \{0,1\}$.
By the probability computation above we know that with probability at least $e^{-q}$, $S^*\subseteq C_0$ and 
$A\setminus S^* \subseteq C_1$. Now we will show that the procedure that prunes $C_0$ and obtains $Q(f)$ does not remove 
any vertices of $S^*$. All the vertices in the set $N(S_2)\cap S$ are contained in $C_1$  and thus there are no edges incident 
to any vertex in $S_2$ in $G[C_0]$. Therefore the only other possibility is that we could remove vertices of 
$S_1\cap C_0$. However, to do so there must be an edge between a vertex in $S_1$ and a vertex in $V_1\cap S$, 
but we know that 
all such neighbors of vertices of $S_1$ are in $C_1$. This shows that with probability at least $2^{-q}$, 
$S^*\subseteq Q(f)$. 


We can boost the success probability of the above random procedure to a constant, by independently repeating the  procedure 
$2^q$ times. Let the random functions obtained while repeating the above procedure 
be $f_{j}$, $j\in \{1,\ldots,2^q\}$ and let $Q(f_j)$ denote the corresponding set obtained after applying the above pruning 
procedure. The probability that  one of the $Q(f_j)$ contains $S^{*}$ is at least 
\[1-\left(1-\frac{1}{2^q}\right)^{2^q} \geq 1-\frac{1}{e}\geq \frac{1}{2}.\]
Thus we obtain a collection $\cal Q$ of subsets of $S$ with the following properties. 
\begin{itemize}
 \item $|{\cal Q}| \leq 2^q$, where ${\cal Q}=\SB Q(f_j) \SM j\in\{1,\ldots,2^q\}\SE$,
\item  every set $Q\in \cal Q$ is an independent set, and 
\item  with probability at least $\frac{1}{2}$, there exists a set $Q\in \cal Q$ such that $S^*\subseteq Q$. 
\end{itemize}

Finally, to derandomize the above procedure we will use Theorem~\ref{propuniversalsets}. We first 
compute a $(|S|,q)$-universal set $\cal F$ with the algorithm described in Theorem~\ref{propuniversalsets} in time 
$O(2^q q^{O(\log q)} |S|\log |S|)$ of size $2^q q^{O(\log q)}\log |S|$. Now every function $f\in {\cal F}$ can be thought of 
as a function from $S$ to $\{0,1\}$. Given this $f$ we obtain $Q(f)$ as described above. Let 
${\cal Q}=\SB Q(f) \SM f\in {\cal F}\SE$. Clearly, $|{\cal Q}|\leq 2^q q^{O(\log q)}\log n$. Now if there exists a set 
$S^*$ of the desired type then the $Q(f)$ corresponding to the function $f\in {\cal F}$, that assigns $0$ to 
every vertex in $S^*$ and $1$ to every vertex in $A\setminus S^*$, has the property that $S^*\subseteq Q(f)$ and $Q(f)$ is an independent set.  
This completes the proof. 
\end{proof}

It is easily seen that finding minimum vertex cover of a 2-subdivided graph is \NP-hard. Indeed, it follows from the \NP-hardness of the \textsc{Vertex Cover} problem on general graphs since: if $G'$ is a 2-subdivision of a graph $G$ with $m$ edges, then $G$ has a vertex
cover of size at most $k$ \myiff $G'$ has a vertex cover of size at most $k+m$.

Thus, Theorems~\ref{thm:hardnessVC} and \ref{thm:fptplsvc} together resolve a question raised by (Krokhin and Marx), who asked for a problem where finding the optimum is hard, 
strict local search is hard, but permissive local search is FPT.

\section{Conclusion}
In this paper we have shown that from the parameterized complexity point of view,
permissive Local Search is  indeed more powerful than the strict Local Search and thus may 
be more desirable.  We have demonstrated this on one example, namely {\sc Vertex Cover},  but it would
be interesting to find a broader set of problems 
where the complexity status 
of the strict and permissive versions of local search differ.  We believe that the results in 
this paper have opened up a complete new direction of research in the domain of parameterized local search, 
which is still in nascent stage. It would be interesting to undertake a similar study for {\sc Feedback Vertex Set}, 
even on planar graphs.

\subsubsection*{Acknowledgments}

All authors acknowledge support from the OeAD 
(Austrian Indian collaboration grant, IN13/2011).  
Serge Gaspers, Sebastian Ordyniak, and Stefan Szeider
acknowledge support from the European Research Council (COMPLEX REASON, 239962).
Serge Gaspers acknowledges support from the
Australian Research Council (DE120101761).
Eun Jung Kim acknowledges support from
the ANR project AGAPE (ANR-09-BLAN-0159).

{

}

%

\end{document}